\documentclass[journal]{IEEEtran}
\usepackage[numbers]{natbib}
\usepackage[utf8]{inputenc}
\usepackage{amsmath,amsthm,amssymb,amsfonts}
\newtheorem{theorem}{Theorem}
\usepackage{url,algorithm, algorithmic, multicol}

\newtheorem{definition}[theorem]{Definition}
\newtheorem{lemma}[theorem]{Lemma}
\newtheorem{corollary}[theorem]{Corollary}

\newcommand\bigO{\mathcal{O}} 
\newcommand\E{\mathbb{E}} 
\newcommand\prob{\mathbb{P}} 
\newcommand\ind{\mathbbold{1}}  
\newcommand\atransition{P^{\text{active}}} 
\newcommand\ptransition{P^{\text{passive}}} 
\newcommand\bandit{\theta} 
\newcommand\bandits{\Theta} 
\newcommand\trueBandit{\theta^{\star}} 
\newcommand\truePolicy{\pi^{\star}} 
\newcommand\prior{Q} 
\newcommand\posterior{Q} 
\newcommand\sumtt{\sum_{t=1}^T} 
\newcommand\sumtl{\sum_{t=1}^L} 
\newcommand\sumin{\sum_{i=1}^N} 
\newcommand\sumll{\sum_{l=1}^{m}} 
\newcommand\history{\mathcal{H}} 
\newcommand\eplen{L} 
\newcommand\epnum{m} 
\newcommand\T{\mathcal{T}} 
\newcommand\initDist{\rho} 

\usepackage{marginnote}

\begin{document}

\title{Screening for an Infectious Disease \\ as a Problem in Stochastic Control}

\author{
 Jakub Mare\v{c}ek~
\thanks{J. Marecek is at the Czech Technical University, Prague, the Czech Republic.
}
\thanks{Manuscript received \today.}
}
\maketitle


\begin{abstract}
There has been much recent interest in screening populations for an infectious disease. Here, we present a stochastic-control model, wherein the optimum screening policy is provably difficult to find, but wherein Thompson sampling has provably optimal performance guarantees in the form of Bayesian regret.
Thompson sampling seems applicable especially to diseases, for which we do not understand the dynamics well, such as to the super-spreading COVID-19.
\end{abstract}

\section{Introduction}

There has been much recent interest in screening populations for a an infectious disease.
In the case of COVID-19, data from contact-tracing apps  \cite[e.g.]{allen2020population,schneider2020flipping,aleta2020modelling,giordano2020modelling,aleta2020modelling}, esp. \cite{liang2020covid}, suggest that mortality is negatively associated with the number of tests performed in the given community, esp. in low-income countries and countries with lower government-effectiveness scores \cite{liang2020covid,brotherhood2020economic}. 
While the association does not imply causation, once the number of tests required by contact tracing exceeds the capacity for performing tests \cite{aleta2020modelling}, or once the contact tracing becomes futile by other means,
the importance of statistical approaches to the allocation of tests to communities \cite{KRETZSCHMAR2020e452,gray2020covid,Muller2020} and, hypothetically,  \cite{yilmaz2020kemeny}, individuals, becomes clear.

Screening for an infectious disease, especially in a pandemic, has two conflicting goals:
one goal is to stop the spread of the disease
and one goal is to understand the spread of the disease 
as precisely as possible.
The former goal may lead to increasing the intensity of testing in communities, where the disease has spread widely. 
The latter goal may lead to uniform sampling from the population, perhaps using tests of limited accuracy,
as exemplified by Slovakia \cite{HOLT2020}, which has tested the entirety of its population with immunoassays over a weekend.
These two goals are conflicting, but their conflict is well understood in Stochastic Control \cite{cesa2006prediction,gittins2011multi}.

Indeed, Stochastic Control could be seen as a field concerned with the balancing the trade-off between ``exploration'' and ``exploitation'' \cite{gittins2011multi}. In exploration, we aim to learn the stochastic processes \cite{dureau2013capturing} involved. In exploitation, one wishes to utilise the current estimates of the stochastic processes involved to optimize a functional, such as the long-run sum of the persons infected.
Notice that the exploitation and exploration does not necessarily map to the short-term and long-term objectives in disease control: 
long-term disease control requires both exploration and exploitation, and hence the balancing the trade-off. 

In contrast, there seems to be a mismatch between the techniques used for screening for an infectious diseases at the moment and the challenges of COVID-19.
Some of the methods for screening for an infectious diseases \cite[e.g.]{Muller2020} are rooted in traditional compartmental models of Epidemiology, while others are data-driven. 
In Computer Science and Statistics, date-driven models are often based on multi-agent techniques \cite[e.g.]{adam2020special,ferguson2020report}, or graph-theoretic considerations \cite[e.g.]{yilmaz2020kemeny}.
In Epidemiology \cite{anderson1992infectious}, there are very many compartmental models \cite{anderson1992infectious}, raging from the simple SIR model \cite{kermack1927contribution}, which recognises three stages,
to SIDARTHE \cite{giordano2020modelling}, which recognises eight stages of infection. 
It is, however, increasingly recognised \cite{lloyd2005superspreading,adam2020clustering} that such models may be of limited utility in screening for novel diseases, which are super-spreading, allow for reinfections, and may overwhelm the healthcare system, for several reasons:
First, it is non-trivial to identify parameters of the compartmental models, until a substantial number of cases of the novel disease is documented in a particular intervention regime.
Second, and more importantly, the compartmental models underestimate the variance of the associated stochastic processes in so-called super-spreading diseases.\footnote{
Statements such as ``offspring distribution of COVID-19 is highly overdispersed'' with $k=0.1$ \cite{endo2020estimating} suggests that “10\% of cases lead to 80\% of the spread” \cite{endo2020estimating}, which is hard to model in the compartmental models.}
Third, many widely studied compartmental models (SI, SIR, SEIR, \ldots, SIDARTHE) do not model reinfections\footnote{
Reinfections are also well documented \cite{duggan2020case,to2020serum}, although their numbers  \cite{ota2020will,gousseff2020clinical,roy2020covid,ccenesiz2020covid} and impact \cite{ccenesiz2020covid} are still unclear.}, not only because of the 
prevalence of diseases with ``immunizing infections,'' but also for technical reasons \cite{keeling2016systematic}.
Fourth, compartmental models that allow for the study of the impact of a test assume that individuals can be effectively isolated, once tested positive, while the overloaded healthcare systems may not be able to prevent further infections by those who have tested positive.
While one can and should make forecasts based on the current models \cite{saad2020immune}, one should also realise that the dynamics are uncertain \cite{saad2020immune} and that 
that modelling the stochastic aspects better \cite{dureau2013capturing} may be beneficial \cite{ting2020digital}.  

Based on a long history of work in Stochastic Control \cite{cesa2006prediction,gittins2011multi},
we present several insights into monitoring the spread of diseases, for which we do not understand the dynamics well, such as the super-spreading COVID-19. 
Overall, our aims are three-fold:
\begin{itemize}
    \item to remove as many assumptions from the screening for an infectious disease as possible, 
    \item to study the computational complexity \cite{arora2009computational} of allocating the budget of tests \cite{emanuel2020fair} independent of any conjectures (e.g., P$\stackrel{?}{=}$NP \cite{arora2009computational}), and
    \item to guarantee optimality of practical algorithms for the same problem.
\end{itemize} 

\section{Our Approach}

\paragraph*{Stochastic Models}
Our first suggestion is to consider the stochastic aspects of the problem explicitly, starting with the fact that the tests are imperfect.\footnote{The probability of detecting disease conditional on the person tested being infected is less than one.
This is true for chest CT and RT-PCR \cite{ai2020correlation} and immunoassays.
Likewise, the probability of detecting disease conditional on the person tested not being infected is larger than zero \cite{ai2020correlation}. 
The probability of a person passing the infection, conditional on them testing positive, is still very much less than one for superspreading \cite{lloyd2005superspreading} diseases.}
Consider the hypothetical situation, where we performed multiple low-accuracy tests of a single person, or perhaps a sequence of tests of increasing accuracy. We should like to consider both the outcomes of the tests for that person, as in their mean, but also some measure of variance of the outcomes for that person, in deciding whether to test further.\footnote{
A classical policy \cite{cesa2006prediction} considers the so-called upper confidence bound (UCB1) based on Hoeffding’s Inequality \cite{cesa2006prediction}. Following $n$ tests in aggregate, out of which $n_i$ tests have been performed on individual $i$ with mean outcome $\mu_i$, individual $i$ will receive an ``index'' $\mu_i + \sqrt{\frac{2\ln{n}}{n_i}}$. Each day,  individuals with the highest indices are chosen for a test, up to the capacity. 
An alternative policy, known as Thompson sampling \cite{thompson1933likelihood}, selects the individual according to the probability that it is optimal, considering some prior.
} 
If at some point, everyone quarantined perfectly, and there were an unlimited capacity to perform the tests, the screening would be reduced to the so-called multi-armed bandit problem (MAB) \cite{gittins2011multi} in Stochastic Control \cite{cesa2006prediction,gittins2011multi}.
(See the Supplementary material for a definition.) If the capacity to perform tests were limited, this would correspond to the combinatorial variant \cite{CESABIANCHI2012,combes2015combinatorial} of the MAB, which is substantially harder \cite{Merlis2020}. 
If, however, the disease spreads, these models are no longer useful and one has to consider the so-called restless bandits \cite{whittle1988restless,gittins2011multi}. 

In particular, in modelling the spread of the disease as restless bandits, the stochastic process could be the positivity rate in a particular region or community, for example with the sampling frequency of a day. One could have several stochastic processes, one for each community.
There are no assumptions on the evolution of the stochastic process, including no assumptions of the independently identically distributed random variables. 


\paragraph*{Computational Complexity}
Consider problem of whom to test given a budget of tests \cite{emanuel2020fair}.
For example in COVID-19, the reported symptom of loss of taste and/or smell was most strongly associated with a positive test result \cite{allen2020population,pierron2020smell}, so in the short term, it may be beneficial to test symptomatic patients. 
It is clear, however, that this is suboptimal in the long run, where one also needs to test asymptomatic individuals \cite{Muller2020} in communities where no (or few) tests have been positive, so far. 
Our second insight is that in the restless-bandit model outlined in the previous paragraph, the problem of whom to test given a budget of tests \cite{emanuel2020fair}, is computationally hard.
In the language of computational complexity, its approximation to any non-trivial factor is complete for polynomial-space Turing machines \cite{Papadimitriou1999}.
This suggests that independent of any unproven conjectures, the problem is as hard as any computation that can be performed using a polynomial amount of space on the Turing machine in \emph{any} amount of time. 
This is based on the well-known complexity results for the restless multi-armed bandit problem \cite{whittle1988restless,weber1990index,gittins2011multi}, under very modest assumptions, as we detail in the Supplementary Material.
 
\paragraph*{Optimality of Thompson sampling}
Our third insight is that there are  (asymptotically) optimal algorithms for the screening problem, despite the complexity results.
Moreover, these optimal algorithms can be as simple as Thompson sampling \cite{thompson1933likelihood,russo2018tutorial}, which is a natural approach, 
wherein one draws a random sample $\bandit_l$ from a prior, 
applies actions
that maximize expected reward considering the sample $\bandit_l$ drawn, observes the outcome, and updates the prior using the observation of the outcome. This is repeated, possibly daily, as suggested in Algorithm~\ref{alg:tsrmab} in the Supplementary Material. 

While the second and third insights may seem contradictory, especially considering that -- until recently -- the best guarantees  \cite{ortner2012regret} for the restless multi-armed bandit problem suggested an $\tilde{\bigO}(\sqrt T)$ bound on the average of the distance between the optimal action and the action chosen by the algorithm (regret) by time $T$ for an algorithm that is intractable in general. 
More recent work \cite{jung2019thompson,jung2019thompson2}  allows for the $\tilde{\bigO}(\sqrt T)$ bound on the regret using Thompson sampling, the classic algorithm \cite{thompson1933likelihood,russo2018tutorial},
both in the special case of binary rewards, where in the testing of individuals there are binary \cite{jung2019thompson} (either an infected person tests positive  with reward 1, or we do not receive any reward), and in a more general episodic case \cite{jung2019thompson2}. 
These are applicable to the two variants of the problem discussed above. 

This can be seen as a complement to the more traditional model-based optimal control \cite{tsay2020modeling,piguillem2020optimal,kruse2020optimal,bin2020fast,kohler2020robust,ma2020optimal,acemoglu2020optimal,birge2020controlling} for the introduction of the restrictions. 

\paragraph*{Making it Practical}
In order to make the optimal algorithms practically relevant, one needs to choose the prior wisely.
The specifics depend, obviously, on the nature of the data available. 
CMU Delphi lab \footnote{
\url{https://cmu-delphi.github.io/delphi-epidata/api/covidcast_signals.html}}, for instance, makes 10 different graphs available, outside of any data from any test-and-trace application.
On such a dataset, for instance, Kemeny-based priors \cite{yilmaz2020kemeny} or highest-degree-first \cite[e.g.]{wilder2020clinical} priors may work well, as documented by clinical trials in other diseases \cite[e.g.]{wilder2020clinical}.
One may also consider extensions drawing on work in reinforcement learning \cite{avrachenkov2020whittle}.

\section{Conclusions}

There has been a substantial amount of proposal as to how to screen for COVID-19, under a variety of strong assumptions. 
We present a very natural approach to removing the assumptions.
This seems particularly useful in diseases, whose dynamics are poorly understood, but may be superspreading and allow for reinfections, e.g., COVID-19.
Within this model, well-known policies from stochastic control come with strong performance guarantees (Bayesian regret bounds) relative to \emph{the best possible} deterministic policies \emph{in hindsight}. Such 
an optimum 
deterministic policy, e.g., based on contact tracing,   are, however,  unknown  \emph{a priori}.
Our guarantees are optimal up to a logarithmic factor.

\clearpage
\bibliography{refs}



\clearpage
\begin{appendix}[Supplementary Material]

In Section \ref{app1} of the Supplementary Material, we introduce the stochastic models involved, following \cite{Papadimitriou1999}.
In Section \ref{appb}, we summarize the guarantees for Thompson sampling, following \cite{jung2019thompson}. 

\subsection{Definitions and Related Work}
\label{app1}

For three decades, one of the best studied problems in applied probability and stochastic analysis has been the restless multi-armed bandit problem, \cite{whittle1988restless,weber1990index,gittins2011multi}.
Formally, in the restless bandits problem, we are given $n$ Markov chains (bandits) $X_i(0), i = 1, \ldots , n, f = 0, 1, .. $, that evolve on a common finite state space $S = {1, .. . , M}$.
We are also given the initial state of each chain. At each time $t$, bandit $i(t)$ is chosen. For $i = i(t), X(t + 1)$ is determined by a transition  matrix $P$.
For every $i 
\not= i(t)$, $X(t + 1)$ is determined by some other transition matrix $Q$. At each time step, we incur a cost $C(t) = c(X_{i(t)}) + \sum_{i \not i(t)} d(X_i(t))$ for some rational-valued functions $c$ and $d$ defined on the state space $S$. Given states of the different bandits,  policy $\pi : S^n \to {1, ... , n}$ decides which bandit should be played next; that is, $i(t) = \pi (X_1(f), \ldots , X_n(t)).$
Its average expected cost is defined as
$$
\lim \sup_{t \to \infty} \frac{1}{T} \sum_{t = 1}^{T}  \mathbb{E}[C(t)],
$$
and we are interested in finding a policy minimizing the average expected cost.
The multi-armed bandit problem is a special case of restless bandits, in which bandits that are not played do not change their state and do not incur any cost, i.e., we have
$Q$ equal to the identity matrix and $d = 0$. 
It is well known that:

\begin{theorem}[Theorem 4 in \cite{Papadimitriou1999}]
Restless bandits are PSPACE-hard.
\end{theorem}

Actually, the proof of Theorem 4 in \cite{Papadimitriou1999} shows that deciding if the optimal reward is non-zero is also PSPACE-hard, hence ruling out any 
algorithm with non-trivial approximation ratio.
Furthermore, the result holds 
\cite{Papadimitriou1999}, even if matrices $P$, $Q$ correspond to one deterministic transition rule for all bandits that are not played and another deterministic transition rule applying to the bandit that is played.

In the general case, \cite{bertsimas2000restless} introduce a hierarchy of $N$ (where $N$ is the number of bandits) increasingly stronger linear programming relaxations, the last of which is exact and corresponds to the (exponential size) formulation of the problem as a Markov decision chain, while the other relaxations provide bounds and are efficiently computed. They  also propose a priority-index heuristic scheduling policy from the solution to the first-order relaxation, where the indices are defined in terms of optimal dual variables. 
Similarly, \cite{ortner2012regret} present strong guarantees, but for an algorithm that is not tractable.

Under stricter assumptions, better guarantees are possible. Under assumptions on the rate of change,
 \cite{besbes2014stochastic} present a framework for reasoning about regret of policies computable in polynomial time. Other assumptions \cite{guha2010approximation,jung2019thompson} also yield guarantees on the approximation ratio. 
Specifically, \cite{jung2019thompson} consider guarantees in the case of the rewards being binary, which is indeed the case in COVID-19 screening.

\begin{algorithm}[t]
	\begin{algorithmic}[1]
	    \STATE \textbf{Input} prior $\prior$, episode length $\eplen$, policy mapping $\mu$
	    \STATE \textbf{Initialize} posterior $\posterior_1 = \prior$, history $\history = \emptyset$
		\FOR {episodes $l = 1, \cdots, \epnum$}
		\STATE Draw a parameter $\bandit_l \sim \posterior_l$ and compute the policy $\pi_l = \mu(\bandit_l)$
		\STATE Set $\history_0 = \emptyset$
		\FOR {$t = 1, \cdots, \eplen$}
		\STATE Select $N$ communities to test in $A_t = \pi_l(t, \history_{t-1})$
		\STATE Evaluate tests to obtain rewards $X_{t, A_t}$ 
		\STATE Update $\history_t$
		\ENDFOR
		\STATE Append $\history_L$ to $\history$ and update posterior distribution $\posterior_{l+1}$ using $\history$
		\ENDFOR
	\end{algorithmic}
	\caption{Thompson sampling for COVID-19 screening, based on \cite{jung2019thompson}}
	\label{alg:tsrmab}
\end{algorithm}

\subsection{The Guarantees}
\label{appb}

Our guarantees are relative to a broad class of benchmark policies including the optimal fixed policy, the myopic policy, or the index-based policy, all of which are:

\begin{definition}[\cite{gittins2011multi,jung2019thompson}]
\label{def:deterministicPolicy}
A deterministic policy $\pi$ takes time index and history $(t, \history_{t-1})$ as an input and outputs a fixed action $A_t = \pi(t, \history_{t-1})$. 
A deterministic policy mapping $\mu$ takes a system parameter $\bandit$ as an input and outputs a deterministic policy $\pi = \mu(\bandit)$. 
\end{definition}

In particular, we bound:

\begin{definition}[\cite{gittins2011multi}]
Regret is:
\begin{equation}
\label{eq:regret}
	R(T;\trueBandit) 
	= \epnum V^{\trueBandit}_{\truePolicy, 1}(\emptyset) - \E_{\trueBandit}\sumtt A_{t} \cdot X_{t}. 
\end{equation}
where the \textit{value function} is:
\begin{equation}
\label{eq:value}
	V^{\bandit}_{\pi, i}(\history) 
	= \E_{\bandit, \pi} [\sum_{j=i}^{\eplen}A_{j} \cdot X_{j} | \history].
\end{equation}
\end{definition}

A variant of the regret, where one assumes we have access to a prior distribution $\prior$ over the set of system parameters $\bandits$:

\begin{definition}[\cite{jung2019thompson}]
Bayesian regret is
\begin{equation*}
	BR(T) = \E_{\trueBandit \sim \prior} R(T;\trueBandit),
\end{equation*}
\end{definition}

The bound is as follows:

\begin{theorem}{\bf(Bayesian regret bound of  Thompson sampling)}
\label{thm:main}
The Bayesian regret of Algorithm \ref{alg:tsrmab} satisfies the following bound
\begin{equation*}
    BR(T) 
    = \bigO (\sqrt{K\eplen^3 N^{3}T \log T})
    = \bigO (\sqrt{mK\eplen^4 N^{3} \log (m\eplen)}).
\end{equation*}
\end{theorem}

\begin{proof}
The proof is by straightforward application of Theorem 1 of \cite{jung2019thompson}.
\end{proof}

It is known \cite{jung2019thompson} that this bound is tight for $\eplen = 1, N=1$.

\end{appendix}
\end{document}